\documentclass[letterpaper,twocolumn,nofootinbib,aps,superscriptaddress]{revtex4}

\usepackage{amsmath}
\usepackage{amssymb}
\usepackage{hyphenat}
\usepackage{amsthm}
\usepackage{bbm}
\usepackage{bm}
\usepackage{upgreek}
\usepackage{hyperref}
\usepackage{graphicx}
\usepackage[figure]{hypcap}
\usepackage{pdftexcmds}

\DeclareFontFamily{U}{cbgreek}{}
\DeclareFontShape{U}{cbgreek}{m}{n}{
        <-6>    grmn0500
        <6-7>   grmn0600
        <7-8>   grmn0700
        <8-9>   grmn0800
        <9-10>  grmn0900
        <10-12> grmn1000
        <12-17> grmn1200
        <17->   grmn1728
      }{}
\DeclareFontShape{U}{cbgreek}{bx}{n}{
        <-6>    grxn0500
        <6-7>   grxn0600
        <7-8>   grxn0700
        <8-9>   grxn0800
        <9-10>  grxn0900
        <10-12> grxn1000
        <12-17> grxn1200
        <17->   grxn1728
      }{}

\makeatletter
\newcommand{\normalorbold}{%
  \ifnum\pdf@strcmp{\math@version}{bold}=\z@ bx\else m\fi
}
\makeatother

\newtheorem{theorem}{Theorem}
\newtheorem{athm}{Theorem}[section]

\newtheorem{propm}[theorem]{Proposition}
\newtheorem{corom}[theorem]{Corollary}

\newtheorem{obs}[athm]{Observation}
\newtheorem{prop}[athm]{Proposition}
\newtheorem{alem}[athm]{Lemma}
\newtheorem{coro}[athm]{Corollary}

\newtheorem*{definition}{Definition}

\newcommand*{\eins}{\ensuremath{\mathbbm 1}}
\def\gbm#1{{\let\lambda\uplambda \let\mu\upmu \let\rho\uprho \let\sigma\upsigma \let\tau\uptau \let\eta\upeta \bm{#1}}}
\newcommand\ltM{\mathrel{\rhd}}

\newcommand\tho{\mathrel{\stackrel{\makebox[0pt]{\mbox{\normalfont\tiny{TO}}}}{\mapsto}}}
\newcommand\ct{\mathrel{\stackrel{\makebox[0pt]{\mbox{\normalfont\tiny{CTO}}}}{\longmapsto}}}

\newcommand*{\bbR}{\mathbb{R}}

\newcommand*{\bbN}{\mathbb{N}}

\newcommand*{\cE}{\mathcal{E}}

\newcommand*{\cL}{\mathcal{L}}
\newcommand*{\cM}{\mathcal{M}}
\newcommand*{\cO}{\mathcal{O}}
\newcommand*{\cR}{\mathcal{R}}
\newcommand*{\cS}{\mathcal{S}}
\newcommand*{\cT}{\mathcal{T}}
\newcommand*{\rA}{\mathrm{A}}
\newcommand*{\rC}{\mathrm{C}}
\newcommand*{\rS}{\mathrm{S}}
\newcommand*{\rX}{\mathrm{X}}
\newcommand*{\rY}{\mathrm{Y}}

\newcommand*{\ket}[1]{\left|#1\right\rangle}
\newcommand*{\bra}[1]{\left\langle #1\right|}
\newcommand*{\proj}[1]{\ket{#1}\bra{#1}}
\newcommand*{\Tr}{\mathrm{Tr}}
\newcommand*{\fr}[2]{\frac{#1}{#2}}
\newcommand*{\tp}{^\mathsf{T}}

\newcommand{\vect}[1]{\mathbf{#1}}

\newcommand{\be}{\begin{equation}}
\newcommand{\ee}{\end{equation}}
\newcommand{\n}{\textendash}

\hypersetup{
    pdftoolbar=true,        
    pdfmenubar=true,        
    pdffitwindow=false,     
    pdfstartview={FitH},    
}

\begin{document}
\title{The resource theory under conditioned thermal operations}
\date{\today}
\author{Varun Narasimhachar}
\email{vnarasim@ucalgary.ca}
\affiliation{Institute for Quantum Science and Technology and Department of Mathematics and Statistics, University of Calgary, 2500 University Drive NW, Calgary, Alberta, Canada T2N 1N4}
\affiliation{School of Physical and Mathematical Sciences, Nanyang Technological University, SPMS-04-01, 21 Nanyang Link, Singapore 637371}
\author{Gilad Gour}
\affiliation{Institute for Quantum Science and Technology and Department of Mathematics and Statistics, University of Calgary, 2500 University Drive NW, Calgary, Alberta, Canada T2N 1N4}
%

\begin{abstract}
The ``thermal operations'' framework developed in past works is used to model the evolution of microscopic quantum systems in contact with thermal baths. Here we extend this model to bipartite devices with one part acting as a control external to the system--bath setup. We define the operations of such hybrid devices as conditioned thermal operations. We examine the resource under these operations, which we call conditional athermality. In the quasiclassical limit, we quantify this resource and find the conditions for its conversion between different forms.
\end{abstract}

\maketitle

\section{Introduction}
In quantum information theory, the unconstrained dynamics of a physical system are mathematically modeled as completely{\hyp}positive trace{\hyp}preserving (CPTP) maps. The actual dynamics may be constrained, e.g.\ by symmetries of the Hamiltonian or practical limitations. Nevertheless, access to systems prepared in some special ``resource'' states can help lift restrictions. Associated with each restricted class of dynamics is the resource that lifts it: entanglement for local operations, reference frames for symmetric dynamics, etc. A \emph{resource theory} (e.g.\ \cite{Asymy9,Reth}) is a formal study of a particular resource, where only a restricted class of operations is ``allowed'' and others forbidden. A ``free state'' is one that can be prepared from scratch using the allowed operations; any non{\hyp}free state is a resource.

Recently, a resource{\hyp}theoretic approach has been taken to non{\hyp}equilibrium thermodynamics of microscopic systems, defining \emph{thermal operations} to model a system's thermal contact with an ideal bath \cite{Nan,Reth,CohRE,catcoh,LJR15,Coh,CST,NG15,Sec,NU,LMP15,FB15,FOR15,AOP15,JA16}. The resource relative to this class of operations is thermal inequilibrium, or \emph{athermality}. In this paper, we define a generalization, \emph{conditioned thermal operations}, wherein the main system undergoes thermal operations \emph{conditioned upon} the state of a control system. We study the theory of the associated resource, which we call \emph{conditional athermality}.

After defining conditional thermal operations in their full generality, we focus on the limiting case where the control system is classical and the main system is quasiclassical (cf.\ Ref.~\cite{Nan,Reth,NU}). This limit is of practical relevance, describing a situation where a microscopic (``quantum'') system in a thermal environment is controlled using macroscopic (``classical'') circuitry. We develop the conditional athermality theory thoroughly in this limit. We first present a method to construct a large class of resourcefulness measures called monotones. Building on the elegant Lorenz curve construction \cite{NU}, we prove that a certain family of monotones within this class provides necessary and sufficient conditions to determine resource convertibility. We also formulate the convertibility problem as an efficiently computable linear program. We then consider large numbers of copies, and find that in the asymptotic limit, all resources are reversiby interconvertible at a rate given by an averaged version of the well-known free energy function. Finally, we discuss the many prospects that lie ahead in the resource theories of athermality and conditional athermality.

\section{Review: thermal operations}
Consider a $d${\hyp}level system S, and let $H$ denote its free Hamiltonian. A \emph{thermal operation} (TO) \cite{Nan} is a state transformation on S effected by (1) introducing an ancilla A, with arbitrarily{\hyp}chosen free Hamiltonian $H_\rA$, prepared in its Gibbs (or thermal) state $\gamma_\rA:=\exp\left(-\beta H_\rA\right)/\Tr\exp\left(-\beta H_\rA\right)$; (2) acting on the combined system SA with a unitary $U_{\rS\rA}$ satisfying $\left[U_{\rS\rA},H+H_\rA\right]=0$ (energy conservation); (3) discarding A. The resulting TO is described by the CPTP map $\cT:\rho\mapsto\Tr_\rA\left[U_{\rS\rA}\left(\rho\otimes\gamma_\rA\right)U_{\rS\rA}^\dagger\right]$, where $\rho$ is an arbitrary initial state of S.

In the resource theory whose allowed operations are TO, the only free state is the thermal state, or Gibbs state, given by
\begin{align}
\gamma:=&\exp\left(-\beta H\right)/\Tr\exp\left(-\beta H\right)\nonumber\\
=&\fr1{Z_\rS}\sum_j\exp\left(-\beta E_j\right)\proj{E_j},
\end{align}
where the $E_j$ are the energy levels (i.e., the eigenvalues of $H$) and $Z_\rS=\sum_j\exp\left(-\beta E_j\right)$ is the evaluation of the system's partition function at temperature $\beta^{-1}$. Deviation from this free state, called athermality, is a resource, as evidenced by its utility in work extraction, refrigeration, and erasure \cite{smarf,Aab,SSP14,CorW,KLOJ15}. The effect of a TO on any state is to bring it closer to $\gamma$. With this background, we are now ready to present our new work.

\section{Conditioned thermal operations}
We now consider a bipartite device consisting of a ``system'' S and a ``control'' C. We define the following:
\begin{definition} A \emph{conditioned thermal operation} (CTO) on the composite $\rS\rC$ is a transformation given by
\be
\cE:\rho_{\rS\rC}\mapsto\sum_{j=1}^n\cT^{(j)}\otimes\cR^j\left(\rho_{\rS\rC}\right),
\ee
where each $\cT^{(j)}$ is a TO and each $\cR^j$ a CP map such that $\sum_j\cR^j$ is TP.
\end{definition}
Note that any CTO belongs to the class of local operations and classical communication (LOCC) with respect to the S\n C partition. If we marginalize over C, the effective transformation of the state of S under a CTO appears like a TO. But note that each of the various conditional TO's $\cT^{(j)}$ acts not on the average marginal state $\rho_\rS$, but on the marginal state that results when the map $\mathrm{id}\otimes\cR^j$ is applied to $\rho_{\rS\rC}$. Therefore, the effective transformation of S is a TO only when $\rho_{\rS\rC}$ is a product state.

\subsection{Free states and resources}
In contrast with TO, there exist an infinite family of free states under CTO, consisting of all states of the form $\gamma\otimes\rho_\rC$, with $\rho_\rC$ arbitrary. If, on the other hand, we consider a wider class of states, satisfying $\Tr_\rC\rho_{\rS\rC}=\gamma$ for some $\rho_{\rS\rC}$, then it is possible that $\Tr_\rC\cE\left(\rho_{\rS\rC}\right)\ne\gamma$ for some CTO's $\cE$. The set of all such ``locally thermal'' states clearly includes all the free states, but also some resource states: those which are locally thermal on S, but contain S\n C correlations. These resource states do not contain any athermality on S relative to uncorrelated systems, but they do \emph{relative to} C, in the sense of Ref.~\cite{RelTher}. We use the term \emph{conditional athermality} for the resource relative to CTO. By definition, CTO cannot create or increase conditional athermality.

\subsection{The role of measurements}
CTO's allow arbitrary measurements on the C system, with the measurement outcome determining the action on S. The outcomes are left ``unread'' from the perspective of external observers, in the sense that a CTO is defined by summing over all possible outcomes. Nevertheless, such measurements possess the power to ``unlock'' the conditional athermality of S relative to C and channel it for useful purposes. Consider, for example, the locally thermal state
\be
\rho_{\rS\rC}=\fr1{Z_\rS}\sum_{j}\exp\left(-\beta E_j\right)\proj{E_j}\otimes\proj j,
\ee
with $\left\{\ket j\right\}$ an orthonormal basis on C. Using thermal operations $\cT^{(j)}$ that effect $\ket{E_j}\mapsto\ket{E_0}$ (the ground state), we can construct a CTO that achieves $\rho_{\rS\rC}\mapsto\proj{E_0}\otimes\sigma_\rC$ (here $\sigma_\rC$ is arbitrary and irrelevant). The classical correlations between S and C have enabled us to change the marginal state of S from the thermal state to a pure energy eigenstate!

Does the above example mean that CTO's trivialize the resource structure of TO's? This is not so: the resource structure induced by CTO's is in fact richer than that under TO's, and subsumes the latter. The TO resource theory can be recovered in its entirety from the CTO theory by considering those instances where both the initial and final state are product states, i.e.\ completely uncorrelated between S and C.

In general, the CTO formalism provides a platform to study the intimate connection between correlations and thermal inequilibrium. Measurements with readout are likely to be even more resourceful in converting correlations to athermality. In general, it is also important to consider the back{\hyp}action of measurements on the control system itself. We leave for future work the study of such measurements and of the deeper connections between correlations and athermality. In the remainder of this paper we develop some of the more basic aspects of the resource theory of conditional athermality.

\begin{figure*}[t!]
\centering
    \includegraphics[width=\textwidth]{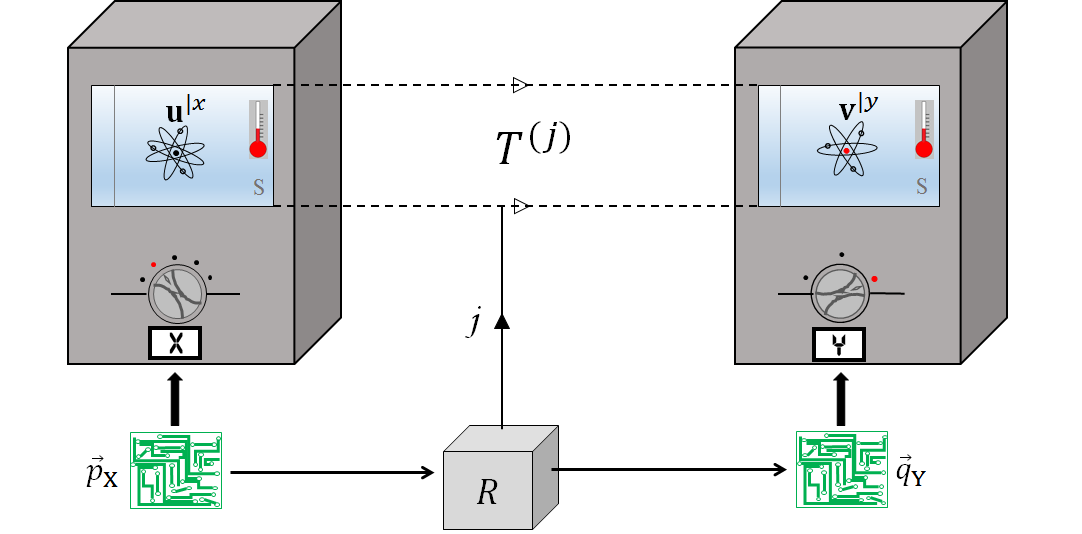}
    \caption[justification=justified]{In a hybrid classical\n quantum thermal device (left), a classical circuit (associated with probability distribution $\vec p_\rX$) determines the classical setting $x$ of the device, which in turn prepares the state $\vect u^{|x}$ on the quantum system S. The rest of the figure shows a schematic of a conditioned thermal operation (CTO): a classical measurement{\hyp}based transformation $R$ takes $\vec p_\rX$ to $\vec q_\rY$, simultaneously outputting a classical value $j$. A corresponding set of thermal operations $T^{(j)}$ determined by $j$ take S from conditional states $\vect u^{|x}$ to $\vect v^{|y}$.}
    \label{Schem}
\end{figure*}
\subsection{The quasiclassical case}
In the remainder, we will develop the resource theory of conditional athermality for the \emph{quasiclassical} regime \cite{Nan,Reth,NU}, where the state of S is a mixture of the eigenstates of its Hamiltonian $H$:
\be\label{qcl}
\rho_\rS=\sum_{i=1}^d{u}_i\proj{i},
\ee
with $\vect u\equiv({u}_1,{u}_2\dots,{u}_d)\tp$ a probability distribution, and the $\ket{i}$ orthonormal eigenvectors of $H$ belonging, respectively, to eigenvalues $E_i$. The Gibbs state $\gamma$ is denoted by the vector $\vect g$. The action of a generic TO on the quasiclassical $\rho_\rS$ of (\ref{qcl}) is effectively a transformation $\vect u\mapsto T\vect u$, with $T$ a stochastic matrix satisfying $T\vect g=\vect g$.

Correspondingly, we will assume that the control is also some classical system X. That is, states of X lack coherence relative to some ``preferred basis'', $\left\{\ket x\right\}$. This can result if the CTO dynamics is much slower than the typical decoherence time scale.

If X has $\ell$ possible settings, it can be prepared in one of those specific settings, or in some probabilistic mixture thereof, represented by a probability distribution $\vec p_\rX\equiv({p}_1,p_2\dots,{p}_\ell)$ (we will represent states of X using row vectors). We can change this state via an arbitrary classical transformation, represented by a $m\times\ell$ row{\hyp}stochastic matrix $R$ mapping $\vec p_\rX\mapsto\vec q_\rY=\vec p_\rX R$ (for clarity, we use different letters to denote the initial and final version of the classical register).

A classically controlled thermal device consists of S and X combined (Fig.~\ref{Schem}, left) in a state represented by a $d\times\ell$ matrix
\begin{align}
U_{\rS\rX}&\equiv\left(p_1\vect u^{|1},p_2\vect u^{|2}\dots,p_\ell\vect u^{|\ell}\right)\nonumber\\
&\equiv\left(\vect u^{1},\vect u^{2}\dots,\vect u^{\ell}\right).
\end{align}
We shorten $p_x\vect u^{|x}$ to $\vect u^{x}$. A generic CTO in the quasiclassical limit takes the form
\be\label{dcto}
U_{\rS\rX}\mapsto\cE\left(U_{\rS\rX}\right)\equiv\sum_{j=1}^nT^{(j)}U_{\rS\rX}{R}^j,
\ee
where each $T^{(j)}$ is a classical TO, and each $R^j$ a sub-stochastic matrix such that $\sum_jR^j$ is stochastic. In the remainder, the term CTO is used in this restricted sense.

\subsection{Measures of conditional athermality}
One way to quantify resources is by constructing resourcefulness measures called \emph{monotones}:
\begin{definition}
A \emph{monotone under CTO}, or \emph{conditional athermality monotone}, is a real{\hyp}valued function $\Phi\left[U_{\rS\rX}\right]$ that does not increase under CTO. That is,
\be\label{dmon}
\Phi\left[\cE\left(U_{\rS\rX}\right)\right]\le\Phi\left[U_{\rS\rX}\right]
\ee
for all quantum\n classical states $U_{\rS\rX}$ and CTO $\cE$.
\end{definition}
Conditional athermality monotones are generalizations of the free energy of classical thermodynamics, in that they can never increase under any allowed evolution. We now find a way to construct a large class of conditional athermality monotones.
\begin{propm}\label{mono}
Let $\phi(\vect u_{\rS})$ be a \emph{convex athermality monotone} on $\rS$. That is, $\phi(T\vect u_{\rS})\le\phi(\vect u_{\rS})$ for all states $\vect u_\rS$ of $\rS$ and all TO ($\vect g${\hyp}preserving column{\hyp}stochastic matrices) $T$, and furthermore, $\phi\left(\alpha\vect u+[1-\alpha]\vect v\right)\le\alpha\phi\left(\vect u\right)+(1-\alpha)\phi\left(\vect v\right)$ for all states $(\vect u,\vect v)$ and $\alpha\in[0,1]$. For every bipartite state $U_{\rS\rX}\equiv\left(p_1\vect u^{|1},p_2\vect u^{|2}\dots,p_\ell\vect u^{|\ell}\right)$, define
\be
\Phi\left[U_{\rS\rX}\right]:=\sum_{x=1}^\ell p_x\phi\left(\vect u^{|x}\right).
\ee
Then $\Phi$ is a conditional athermality monotone.
\end{propm}
In the following, we will find conditions for resource interconvertibility, in terms of a family of such monotones. More such monotones are likely to be involved in the conditions for \emph{catalytic} conversion, which we will not study in this work.

\subsection{Single{\hyp}copy conditional athermality conversion}
Given one copy each of two arbitrary states $U_{\rS\rX}$ and $V_{\rS\rX}$, how do we determine if $U_{\rS\rX}\ct V_{\rS\rX}$, i.e., if there exists a CTO $\cE$ such that $V_{\rS\rX}=\cE\left(U_{\rS\rX}\right)$?
\begin{figure*}[t]
\centering
    \includegraphics[width=\textwidth]{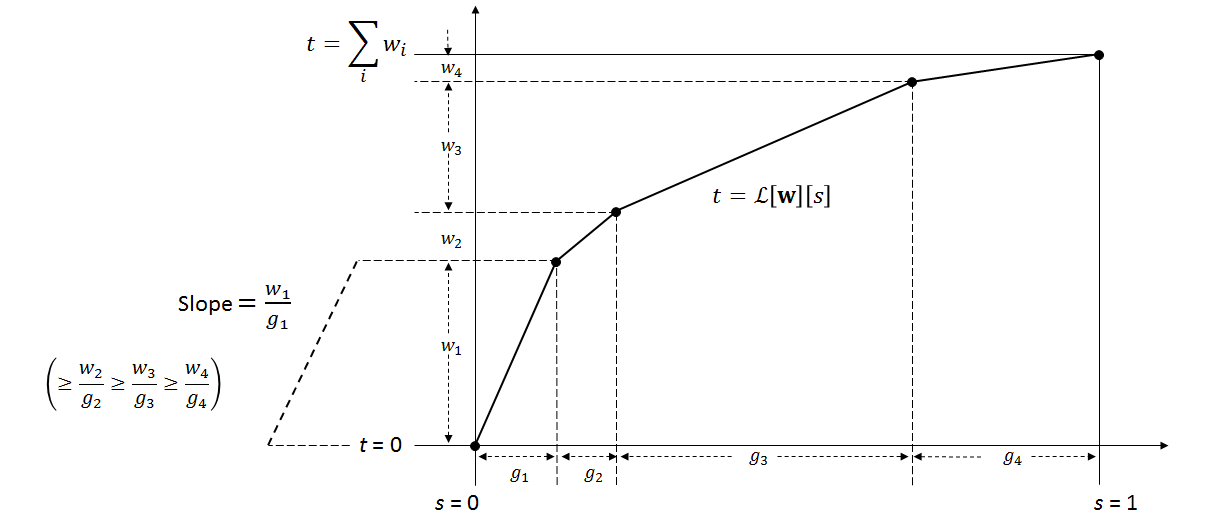}
    \caption[justification=justified]{The Lorenz curve construction for a sub-normalized vector $\vect w$ on a system with Gibbs state $\vect g$: First index the energy eigenstates such that ${w}_1/g_1\ge{w}_2/g_2\dots\ge{w}_d/g_d$. Then connect the points $(0,0)$, $(g_1,{w}_1)$, $(g_1+g_2,{w}_1+{w}_2)\dots$, $(1,1)$ with line segments to obtain the Lorenz curve $\cL[\vect w]$. The curve defines a function $\cL[\vect w](s)$, which we call the Lorenz function of $\vect w$.}
    \label{LorC}
\end{figure*}

Ref.~\cite{Nan} found that a state transformation $\vect u\tho\vect v$ is possible if and only if the \emph{Lorenz curve} (see Fig.~\ref{LorC}) of $\vect u$ is nowhere below that of $\vect v$:
\begin{equation}\label{tolor}
\cL\left[\vect u\right](s)\ge\cL\left[\vect v\right](s)\quad\forall s\in[0,1].
\end{equation}
This condition is described as ``$\vect u$ \emph{thermo{\hyp}majorizes} $\vect v$''. Combining this condition with the structure of CTO defined in Eq.~\ref{dcto}, we prove the following (details in the appendix):
\begin{propm}\label{plqm}
For $U$ and $V$ as defined above, $U\ct V$ if and only if there exists an $\ell\times m$ row{\hyp}stochastic matrix ${R}\equiv\left(r_{xy}\right)$, such that for each $y\in\{1,2\dots,m\}$,
\be\label{lorcom}
\sum_{x=1}^\ell r_{xy}\cL\left[\vect u^{x}\right](s)\ge\cL\left[\vect v^{y}\right](s)\quad\forall s\in[0,1].
\ee
Note that the Lorenz functions appearing above are all sub{\hyp}normalized by virtue of the sub{\hyp}normalization of the $\vect u^x$ and the $\vect v^y$.
\end{propm}
The case $s=1$ in the family of inequalities~(\ref{lorcom}) implies $\vec q_\rY=\vec p_\rX R$, as expected of the marginal on the classical register. On the other hand, summing over $y$ gives us the condition $\sum_y\cL\left[\vect v^{y}\right](s)\le\sum_x\cL\left[\vect u^{x}\right](s)$ for all $s$, i.e., that the quantum part of $U$ must thermo{\hyp}majorize that of $V$ on average (where the average is taken \emph{after} the evaluation of the Lorenz function).
%

The condition (\ref{lorcom}) runs over $s\in[0,1]$, but we don't really need to check for all values of $s$. Note that $\cL(\vect w)$ for any $\vect w$ is continuous and piecewise linear, with at most $d-1$ ``bends''. In addition, it is also concave and monotonously non{\hyp}decreasing. Therefore in order to determine whether a given $\vect u$ thermo{\hyp}majorizes $\vect v$, it suffices to compare their Lorenz curves only at the abscissae where $\cL\left[\vect v\right]$ bends.

Let us now reconsider the convertibility question $U\ct V$. For any given $y$, there are at most $(d-1)$ bends in $\cL\left[\vect v^{y}\right]$; therefore, all $y$'s considered, there are up to $m(d-1)$. Define $D$ such that $(D-1)$ is the total number of distinct bend abscissae ($D\le m[d-1]+1$), which we shall call $s_1<s_2\dots<s_{D-1}$. In addition, let $s_0:=0$ and $s_D:=1$. Now define the $D\times\ell$ matrix $P$ and the $D\times m$ matrix $Q$:
\begin{align}\label{uvpq}
p_{ix}&=\cL\left[\vect u^{x}\right](s_i)-\cL\left[\vect u^{x}\right](s_{i-1});\nonumber\\
q_{iy}&=\cL\left[\vect v^{y}\right](s_i)-\cL\left[\vect v^{y}\right](s_{i-1}).
\end{align}
Note that $P$ and $Q$ are normalized bipartite probability distributions. Also, every pair $(U,V)$ uniquely determines a pair $(P,Q)$. In fact, $Q$ is uniquely determined by $V$; however, $P$ depends on both $U$ and $V$, because we chose the $s_i$ based on the Lorenz curves of the $\vect v^y$. In terms of $P$ and $Q$, Proposition~\ref{plqm} can be recast as follows:
\begin{propm}\label{lempq}
For any pair $(U,V)$, define $({P},{Q})$ as above. Then,
\be\label{lprq}
U\ct V\;\Longleftrightarrow\;\exists\textnormal{ row{\hyp}stochastic }{R}:L{P}{R}\ge L{Q},
\ee
where the inequality is entriwise, and $L$ is the $D\times D$ lower{\hyp}triangular matrix with all diagonal and lower elements equalling 1.
\end{propm}
\begin{corom}\label{stec}
In the case $\ell=1$, i.e. when $U\equiv\vect u$ is just an athermality resource with a trivial classical register, $U\ct V$ if and only if
\be\label{estec}
\cL[{\vect p}](s_i)\ge\cL\left[{\vect q^{|y}}\right](s_i)\quad\forall y\in\{1\dots,m\},i\in\{1\dots,D\}.
\ee
\end{corom}
\begin{corom}\label{corom}
In the case $m=1$, i.e., when $V\equiv\vect v$ and $U\equiv\left(p_1\vect u^{|1},p_2\vect u^{|2}\dots,p_\ell\vect u^{|\ell}\right)$, $U\ct V$ if and only if
\be\label{ecorom}
\sum_{x=1}^\ell p_x\cL\left[{\vect p^{|x}}\right](s_i)\ge \cL[{\vect q}](s_i)\quad\forall i\in\{1,2\dots,D\}.
\ee
\end{corom}
\begin{corom}\label{coro5}
Given athermality resources $(\vect u,\vect v)$ such that $\vect u\tho\vect v$ is possible, $U\equiv\left(p\vect u,[1-p]\vect g\right)$ can be converted to $V\equiv\vect v$ by CTO if and only if $p\ge p_{\min}$, where
\be
p_{\min}=\max_{i\in\{1\dots,D-1\}}\left[\fr{\cL[\vect v](s_i)-s_i}{\cL[\vect u](s_i)-s_i}\right].
\ee
\end{corom}
We see from the corollaries that the state{\hyp}to{\hyp}ensemble case ($\ell=1$) reduces to several independent instances of athermality resource convertibility. Note that this special case is not the same as the probabilistic conversion problems considered in Ref.~\cite{AOP15}. On the other hand, in the ensemble{\hyp}to{\hyp}state case ($m=1$), only the ``average resourcefulness'' of the initial ensemble matters, although the classical register is still important because in general $\sum_xp_x\cL\left[\vect u^{|x}\right](s)\ge\cL\left[\sum_xp_x\vect u^{|x}\right](s)$.

The results of corollaries \ref{stec}, \ref{corom}, and \ref{coro5} are mathematically analogous to corresponding results about probabilistic conversion of pure entanglement resources \cite{SLOCC,Vidal}, with the roles of initial and final states reversed and majorization \cite{MOA} replaced by thermo{\hyp}majorization.

Proposition~\ref{lempq} implies that every instance of the CTO convertibility problem is the feasibility problem of a linear program, which can be solved efficiently using state{\hyp}of{\hyp}the{\hyp}art computers. In this form, the relation $\ct$ becomes very similar to the ``conditional majorization'' relation defined in recent work \cite{CUP15} on the uncertainty principle in the presence of a memory (``conditional uncertainty''). For two $D${\hyp}dimensional vectors $\vect u$ and $\vect v$, $L\vect u\ge L\vect v$ is equivalent to the existence of a $D\times D$ lower{\hyp}triangular column{\hyp}stochastic (LTCS) matrix $\Theta$ such that $\vect v=\Theta\vect u$. This condition has been called \emph{lower{\hyp}triangular majorization} (LT majorization) in Refs.~\cite{MOA,NG15}. By applying methods of convex geometry and the properties of LT majorization (details in the appendix), we translate the conditional athermality convertibility condition (\ref{lprq}) to a family of no{\hyp}go conditions parametrized by matrices from the set
\be
\bbR^{D\times m}_{+,1,\downarrow}=\left\{A\in\bbR^{D\times m}_{+,1}:\forall j,\;a_{1j}\ge a_{2j}\dots\ge a_{Dj}\right\},
\ee
where $\bbR^{D\times m}_{+,1}$ is the set of all $D\times m$ joint distributions. For $A\equiv\left(\vect a^1,\vect a^2\dots,\vect a^m\right)\in\bbR^{D\times m}_{+,1,\downarrow}$, define the sublinear functional $\omega_A:\bbR_+^D\to\bbR$ through
\be\label{supf}
\omega_A(\vect w):=\max_{z\in\{1,2\dots,m\}}\vect a^{z}\cdot\vect w.
\ee
Similar in spirit to the monotones in Proposition~\ref{mono}, define
\be
\Omega_A\left[U,V\right]:=\sum_{x=1}^\ell p_x\omega_A\left(\vect p^{|x}\right)-\sum_{y=1}^mq_y\omega_A\left(\vect q^{|y}\right),
\ee
with $\vect p^{|x}$ and $\vect p^{|x}$ as defined by (\ref{uvpq}). Then,
\begin{theorem}\label{thmon}
For an arbitrary pair of conditional athermality resources $(U,V)$, let $({P},{Q})$ be defined as in Proposition~\ref{lempq}. Then, $U\ct V$ if and only if, for all matrices $A\in\bbR^{D\times m}_{+,1,\downarrow}$,
\be
\Omega_A\left[U,V\right]\ge0.
\ee
\end{theorem}
This result provides sufficient conditions for resource conversion through an efficiently computable family of functions. Note that both $P$ and $Q$ depend on the values of the $s_i$, which in turn depend on the Lorenz curves of the target states $\vect v^{|y}$. Consequently, the quantities $\omega_A(\vect p^{x})$ depend on both the source and target, and so $\Omega_A$ is not a monotone. But our $V${\hyp}dependent choice of $s_i$ was motivated by the goal to minimize the complexity of the problem. In the appendix we will use essentially the same method to construct a sufficient family of monotones.

\subsection{Asymptotic conversion}
The asymptotic limit pertains to the following problem: Given a pair $(U,V)$ of conditional athermality resources, what is the optimal rate $\left(n_2/n_1\right)$, as $n_1\to\infty$, such that $U^{\otimes n_1}\ct V^{\otimes n_2}$ (allowing a conversion error that vanishes in the limit)? Applying ideas from the theory of asymptotic equipartition and previous results \cite{Reth} about athermality, we prove the following (details in appendix):
\begin{propm}\label{asym}
Asymptotic conversion $U^{\otimes n_1}\ct V^{\otimes n_2}$ can be carried out \emph{reversibly}, at the optimal rate
\be
\lim_{n_1\to\infty}\fr{n_2}{n_1}=\fr{{f}(U)}{{f}(V)},
\ee
where for $U\equiv\left(p_1\vect u^{|1},p_2\vect u^{|2}\dots,p_\ell\vect u^{|\ell}\right)$,
\be
{f}(U):=\sum_{x=1}^\ell p_xF_\beta\left(\vect u^{|x}\right),
\ee
with $F_\beta(\vect u):=\sum_{i=1}^d{u}_i\left(E_i+\beta^{-1}\ln{u}_i\right)$ the \emph{free energy function} of classical thermodynamics.
\end{propm}
Whereas determining the convertibility of finite resources requires the calculation of infinitely many functions or searching through infinitely many possibilities, only one easily{\hyp}computable function suffices in the asymptotic limit. This function, namely the averaged free energy, can therefore be seen as a standard measure of asymptotic conditional athermality resourcefulness. Consequently, while the $\ct$ relation is a partial preorder in general, it turns into a \emph{total} preorder in the asymptotic limit: even if $U$ and $V$ are incomparable resources in finite numbers of copies, $U^{\otimes n/{f}(U)}$ and $V^{\otimes n/{f}(V)}$ become equally resourceful as $n\to\infty$. For this reason, the resource conversion is reversible in the asymptotic limit, unlike in the finite case.

\section{Conclusion}
We extended the existing formalism of thermal operations (TO) and the associated athermality resource theory to characterize the thermodynamic transitions achievable on a microscopic thermal system controlled through another system external to the thermal contact. Using a formalism with an explicitly bipartite system, we extended the TO model to define \emph{conditioned thermal operations} (CTO). We defined the resource under CTO as \emph{conditional athermality}, and identified some of its key properties.

In the quasiclassical limit of CTO, we developed a thorough resource theory of conditional athermality. We first found a general recipe for constructing measures of conditional athermality. We then derived necessary and sufficient conditions for single{\hyp}copy resource convertibility, both in terms of a family of efficiently{\hyp}computable monotones and as a linear program. As corollaries, we found the conditions for state{\hyp}to{\hyp}ensemble and ensemble{\hyp}to{\hyp}state conversion of athermality resources. These conditions are very similar to analogous problems for pure bipartite states under local operations and classical communication (LOCC) \cite{SLOCC,Vidal}, with the roles of initial and final states reversed and thermo{\hyp}majorization replaced by majorization. Finally, we found that the asymptotic limit of the conditional athermality resource theory is reversible. The value of every resource in this limit is determined by the classical free energy averaged over the ensemble.

At first glance, the state{\hyp}to{\hyp}ensemble case has a similar flavour to the work Ref.~\cite{AOP15}, whose authors found the greatest probability with which a given athermality resource conversion can be achieved under TO. But there are important differences: in the ``heralded probabilistic conversion'' of Ref.~\cite{AOP15}, (1) the classical register is also in the thermal environment; (2) an additional ancillary resource is allowed to particiate (without getting consumed); and (3) measurements are allowed on the thermally{\hyp}evolving systems (although the costs of the measurements are conscientiously tracked). In the simpler, unheralded case, their formalism does not involve a classical register indicating the states on the quantum system. Therefore, conditional athermality convertibility is a stricter condition than the unheralded convertibility considered in Ref.~\cite{AOP15}. Its exact relation with heralded convertibility, as well as the incorporation of measurements into the formalism, is a topic we hope to probe in the future.

Development of the athermality and conditional athermality resource theories away from the quasiclassical limit is a topic of ongoing research. Also part of future work is the use of the CTO formalism to probe the exact relations between correlations and athermality. The results of this paper barely scratch the surface, but are rather intended to be demonstrative of the richness of the conditional athermality resource theory. Further development of this resource theory would also consider catalytic conversion and approximate conversion. Going beyond the CTO formalism, it would also be useful to consider measurements with readout and back{\hyp}action.

\section*{Acknowledgments}
We thank Mile Gu, Jiajun Ma, Iman Marvian, and Jayne Thompson for helpful discussions. We acknowledge support from the Natural Sciences and Engineering Research Council of Canada (NSERC). VN acknowledges financial support from the Ministry of Education of Singapore, the National Research Foundation (NRF), NRF-Fellowship (Reference No: NRF-NRFF2016-02) and the John Templeton Foundation (Grant No 54914).


\appendix

\setcounter{secnumdepth}{2}
\numberwithin{equation}{section}
\makeatletter
\renewcommand{\theequation}{\thesection.\arabic{equation}}

\section{Review: Thermo-majorization}
The relation of thermo-majorization is defined through a plane-geometric construction called the Lorenz curve \cite{Nan}. In the quasiclassical resource theory of thermal operations, given a fixed ambient inverse temperature $\beta$, the Lorenz curve (Fig.~\ref{LorC} of main matter) is a function of the state $\vect u$ and the Hamiltonian $H$. Since we assume $H$ fixed, the Lorenz curve is just a function of $\vect u$. It captures the way in which the state $\vect u$ differs from the Gibbs state $\vect g$ given by $g_i:=\exp\left(-\beta E_i\right)$.
\begin{definition}[Lorenz curve]
For a vector $\vect u$ with nonnegative components, index the standard energy eigenstates in such a way that
\begin{equation}\label{thL}
\fr{{u}_1}{g_1}\ge\fr{{u}_2}{g_2}\dots\ge\fr{{u}_d}{g_d}.
\end{equation}
Then, the \emph{Lorenz curve $\cL[\vect u]$} is a curve on the truncated plane $[0,1]\times[0,1]$, constructed as follows. First, mark the points
\be
(0,0),(g_1,{u}_1),(g_1+g_2,{u}_1+{u}_2)\dots,(1,u_1+u_2\dots+u_d).
\ee
$\cL[\vect u]$ is obtained by joining (with straight-line segments) adjacent pairs in this sequence of $(d+1)$ points. For every abscissa $s\in[0,1]$, there is a unique point $(s,t)$ on the curve $\cL[\vect u]$. Therefore, we can express the curve by specifying $t$ as a function of $s$. We use the same notation for this function (``the Lorenz function of $\vect u$'') as for the curve itself:
\be
t=\cL[\vect u](s).
\ee
\end{definition}
Note that sub-normalized vectors also have Lorenz curves and functions as per our definition. For normalized vectors (distributions), which represent states of S in our formalism, we define the following relation:
\begin{definition}[Thermo-majorization]
Given states $\vect u$ and $\vect v$, $\vect u$ is said to \emph{thermo-majorize} $\vect v$ if the thermo-Lorenz curve $\cL[\vect u]$ is nowhere below $\cL[\vect v]$. Throughout the supplementary material, we will abbreviate this condition as $\cL[\vect u]\ge\cL[\vect v]$. If $\cL[\vect u]\le\cL[\vect v]$ also holds, we will say $\cL[\vect u]=\cL[\vect v]$.
\end{definition}
Fig.~\ref{ThMaj} illustrates thermo-majorization through examples. In Ref.~\cite{Nan}, it was proved that $\vect u\tho\vect v$ if and only if $\cL[\vect u]\ge\cL[\vect v]$. Before we proceed, we note some general properties of the thermo-Lorenz construction.
\begin{figure*}[t]
    \includegraphics[width=\textwidth]{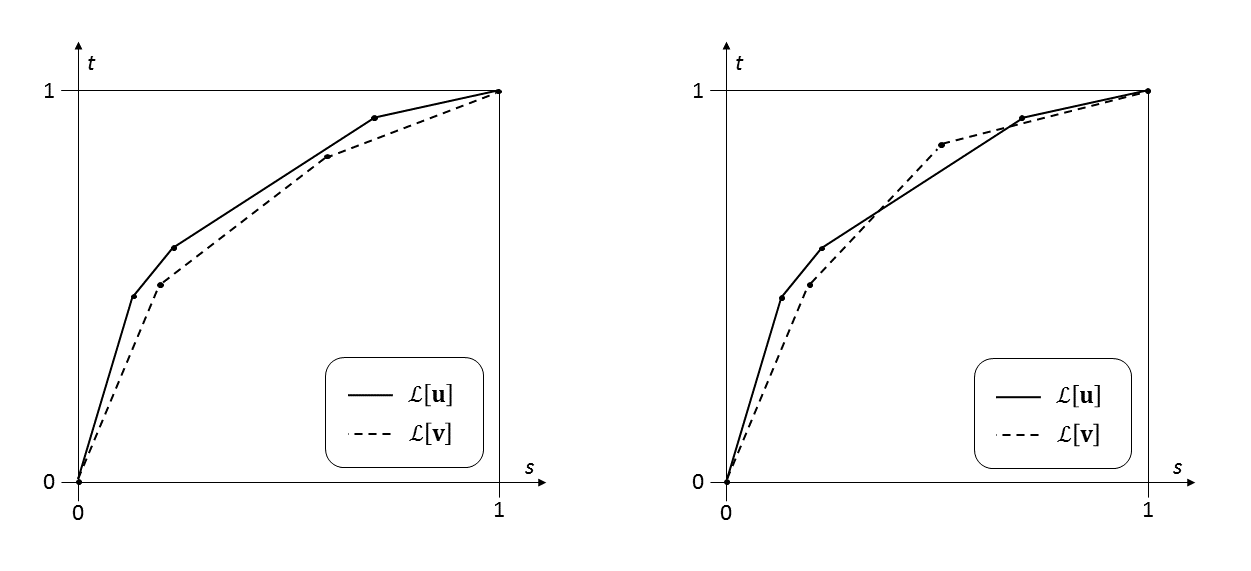}
    \caption[justification=justified]{Thermo-majorization: On the left, $\vect u$ thermo-majorizes $\vect v$; on the right, neither state thermo-majorizes the other.}
    \label{ThMaj}
\end{figure*}

The indexing (\ref{thL}) ensures that the curve is concave. It is also continuous and piecewise linear by construction, with at most $(d-1)$ bends for a state on a $d$-dimensional system. We can use a succinct representation wherein we describe a Lorenz curve by specifying the coordinates of its bends. The construction is also efficiently computable from matrix representations of the density operator and the Hamiltonian. Two states $(\vect u_1,H_1)$ and $(\vect u_2,H_2)$ (momentarily allowing the Hamiltonian to vary) that have the same thermo-Lorenz curve can be considered equivalent under TO, or \emph{TO-equivalent}, because either state can be converted to the other. In order to ascertain whether $\cL[\vect u]\ge\cL[\vect v]$, we have to compare the two curves only at the points where $\cL[\vect v]$ bends. Consequently, the predicate ``$\cL[\vect u]\ge\cL[\vect v]$" is equivalent to a number of scalar inequalities equal to the number of bends in $\cL[\vect v]$.

\section{Proofs of CTO results}
\noindent Let us begin by recalling the definition of CTO:
\begin{definition}[CTO]
Given a $d\times\ell$ joint state $U\equiv U_{\rS\rX}$, a CTO is an operation determined by an indexed family $\left(T^{(j)}\right)_{j\in\{1,2\dots,n\}}$ of TO ($\vect g$-preserving $d\times d$ column-stochastic matrices), and a set of corresponding $\ell\times m$ row-substochastic matrices $\left({R}^j\right)_{j\in\{1,2\dots,n\}}$, such that ${R}:=\sum_{j=1}^n{R}^j$ is row-stochastic. The effect of the operation on $U$ is
\be\label{ctj}
U\mapsto\sum_{j=1}^nT^{(j)}U{R}^j.
\ee
\end{definition}
\begin{obs}\label{jxy}
Without loss of generality, we can choose a decomposition where the index $j$ is replaced by pairs $(x,y)$, with $x\in\{1,2\dots,\ell\}$ and $y\in\{1,2\dots,m\}$. For, given a decomposition (\ref{ctj}), let ${R}_{xy}$ denote the $(x,y)^\textnormal{th}$ element of the matrix ${R}$. For each pair $(x,y)$, let $\tilde{R}^{x,y}$ denote the $\ell\times m$ matrix that has ${R}_{xy}$ as its $(x,y)$ element and zeroes everywhere else. Now define the family $\left(\tilde T^{(x,y)}\right)_{x,y}$ through
\be
\tilde T^{(x,y)}:=\fr1{{R}_{xy}}\sum_{j=1}^n{R}^j_{xy}T^{(j)}.
\ee
Each such matrix is a convex combination of TO, and is therefore itself a TO. Moreover, $\sum_{x,y}\tilde{R}^{x,y}={R}$. It can be verified that
\be
\sum_{x,y}\tilde T^{(x,y)}(\cdot)\tilde{R}^{x,y}=\sum_{j=1}^nT^{(j)}(\cdot){R}^j,
\ee
and so the LHS defines an alternative decomposition of the same CTO.\qed
\end{obs}
\begin{coro}\label{ptq}
For $U\equiv\left(p_1\vect u^{|1},p_2\vect u^{|2}\dots,p_\ell\vect u^{|\ell}\right)\equiv\left(\vect u^{1},\vect u^{2}\dots,\vect u^{\ell}\right)$ and $V\equiv\left(q_1\vect v^{|1},q_2\vect v^{|2}\dots,q_m\vect v^{|m}\right)\equiv\left(\vect v^{1},\vect v^{2}\dots,\vect v^{m}\right)$, $U\ct V$ if and only if there exists a family $\left(T^{(x,y)}\right)_{x,y}$ of TO, and an $\ell\times m$ row-stochastic matrix ${R}$, such that for each $y\in\{1,2\dots,m\}$,
\be
\vect v^{|y}=\sum_{x=1}^\ell r_{x|y}T^{(x,y)}\vect u^{|x},
\ee
where $r_{x|y}:=p_x{R}_{xy}/q_y$.
\end{coro}
The conditions in the above result include existential clauses invoking several objects: the matrix ${R}$, and the family of TO $T^{(x,y)}$. We will eventually reduce the conditions to a form where only the existence of ${R}$ is invoked, but first we prove Proposition~\ref{mono} about CTO monotones:
\begin{prop}[Proposition~\ref{mono} of main text] Let $\phi(\vect u_{\rS})$ be a \emph{convex TO monotone} on $\rS$. That is, $\phi(T\vect u_{\rS})\le\phi(\vect u_{\rS})$ for all states $\vect u_\rS$ of $\rS$ and all TO ($\vect g$-preserving column-stochastic matrices) $T$, and furthermore, $\phi\left(\alpha\vect u+[1-\alpha]\vect v\right)\le\alpha\phi\left(\vect u\right)+(1-\alpha)\phi\left(\vect v\right)$ for all states $(\vect u,\vect v)$ and $\alpha\in[0,1]$. For every bipartite state $U_{\rS\rX}\equiv\left(p_1\vect u^{|1},p_2\vect u^{|2}\dots,p_\ell\vect u^{|\ell}\right)$, define
\be
\Phi\left[U_{\rS\rX}\right]:=\sum_{x=1}^\ell p_x\phi\left(\vect u^{|x}\right).
\ee
Then $\Phi$ is a CTO monotone.
\end{prop}
\begin{proof}
Suppose $U\equiv\left(p_1\vect u^{|1},p_2\vect u^{|2}\dots,p_\ell\vect u^{|\ell}\right)$ and $V\equiv\left(q_1\vect v^{|1},q_2\vect v^{|2}\dots,q_m\vect v^{|m}\right)$ such that $U\ct V$. By Corollary~\ref{ptq} there exists a family $\left(T^{(x,y)}\right)_{x,y}$ of TO, and an $\ell\times m$ row-stochastic matrix ${R}$, such that for each $y\in\{1,2\dots,m\}$,
\be
\vect v^{|y}=\sum_{x=1}^\ell r_{x|y}T^{(x,y)}\vect u^{|x},
\ee
where $r_{x|y}:=p_x{r}_{xy}/q_y$. Note that each $\vect q^{|y}$ is a convex combination of various $T^{(x,y)}\vect u^{|x}$. We therefore have, for a convex TO monotone $\phi$,
\begin{align}
\sum_{y=1}^{m}q_{y}\phi\left(\vect v^{|y}\right)&=\sum_{y=1}^{m}q_{y}\phi\left(\sum_{x=1}^{\ell}r_{x|y}T^{(x,y)}\vect u^{|x}\right)\nonumber\\
&\le\sum_{x=1}^{\ell}\sum_{y=1}^{m}q_{y}r_{x|y}\phi\left(T^{(x,y)}\vect u^{|x}\right)\nonumber\\
&=\sum_{x=1}^{\ell}p_x\sum_{y=1}^{m}r_{xy}\phi\left(T^{(x,y)}\vect u^{|x}\right)\nonumber\\
&\le\sum_{x=1}^{\ell}p_x\sum_{y=1}^{m}r_{xy}\phi\left(\vect u^{|x}\right)\nonumber\\
&=\sum_{x=1}^{\ell}p_x\phi\left(\vect p^{|x}\right)\sum_{y=1}^{m}{r}_{xy}\nonumber\\
&=\sum_{x=1}^{\ell}p_x\phi\left(\vect p^{|x}\right).
\end{align}
The first inequality follows from the convexity of $\phi$, and the second one from its monotonicity under TO.
\end{proof}
We now work towards our main result by proving some useful results about Lorenz curves and thermo-majorization.
\begin{alem}\label{lemlc}
Under a fixed Hamiltonian, the Lorenz curve is a convex property of the state. That is,
\be
\cL\left[\sum_jr_j\vect w^j\right]\le\sum_jr_j\cL\left[\vect w^j\right]
\ee
for any probability distribution $\vect r$ and collection $(\vect w^j)$ of states.
\end{alem}
\begin{proof}
For the purpose of this proof we will have to consider systems with different Hamiltonians, and therefore Lorenz curves as functions of both states and Hamiltonians. First consider the case of a Hamiltonian $H$ whose associated Gibbs distribution has components $g_i=\exp\left(-\beta E_i\right)/Z$ that are mutually rational. Let's call such a Hamiltonian \emph{Gibbs-rational}. We first find the greatest common divisor $g$ of all the $g_i$'s, and define the integer
\be
\tilde d:=\fr{\sum_{i=1}^dg_i}g.
\ee
For an arbitrary state $\vect w$ under $H$, we can always find a $\tilde d$-dimensional state $\tilde{\vect w}$ whose Lorenz curve $\cL[\tilde{\vect w},\tilde H]$ under the trivial Hamiltonian $\tilde H:=\eins_{\tilde d}$ is identical to $\cL[\vect w,H]$ (Fig.~\ref{GRat}). The components of $\tilde{\vect w}$ are of the form $\tilde{w}_i:={w}_ig/g_i$, with each $\tilde{w}_i$ repeated $g_i/g$ times. This construction commutes with convex combination. That is, if $\vect w=\sum_jr_j\vect w^j$, then the $\tilde d$-dimensional states $\tilde{\vect w}$ and $\tilde{\vect w}^j$ constructed in the above manner satisfy $\tilde{\vect w}=\sum_jr_j\tilde{\vect w}^j$.

Now, since the $\tilde d$-dimensional Gibbs state is just the uniform distribution, the ordinates of the bends in the Lorenz curve of a state $\tilde{\vect w}$ are given by the partial sums of $\tilde{\vect w}^\downarrow$, the vector obtained by arranging the components of $\tilde{\vect w}$ in nonincreasing order: $\tilde v_1=\tilde{w}_1^\downarrow$, $\tilde v_2=\tilde{w}_1^\downarrow+\tilde{w}_2^\downarrow$, etc.

Under a trivial Hamiltonian (i.e., one whose Gibbs distribution is uniform), the convexity of the Lorenz curve as a function of the state follows from the convexity of these partial sums. Using our construction, the property extends to any Gibbs-rational Hamiltonian. Since a general Hamiltonian can be approximated arbitrarily well by a Gibbs-rational Hamiltonian, the lemma follows.
\end{proof}
\begin{figure*}[t]
    \includegraphics[width=\textwidth]{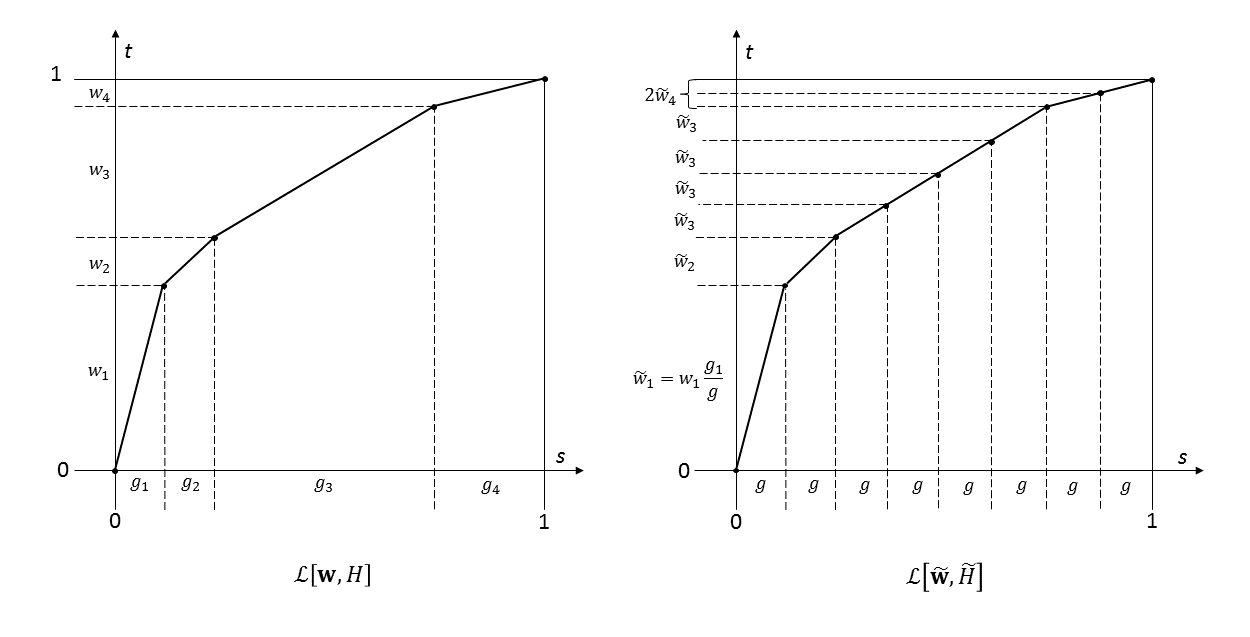}
    \caption[justification=justified]{If $H$ is a Hamiltonian whose Gibbs state $\vect g$ has mutually rational components, then for every $\vect w$ under $H$ we can find a higher-dimensional state $\tilde{\vect w}$ under a Hamiltonian $\tilde H$ whose Gibbs state is uniform, such that $\cL[\tilde{\vect w},\tilde H]=\cL[\vect w,H]$.}
    \label{GRat}
\end{figure*}
\begin{alem}\label{lemlm}
For any finite collection $\left(\vect w^j\right)_{j\in\{1\dots n\}}$ of normalized states on a $d$-dimensional system under Hamiltonian $H$, there exist $D$-dimensional states $\left(\tilde{\vect w}^j\right)_{j\in\{1\dots n\}}$ under some Hamiltonian $\tilde H$, where $D=\cO(nd)$, such that:
\begin{enumerate}
\item For each $j\in\{1\dots n\}$,
\be\label{loreq}
\cL\left[\tilde{\vect w}^j,\tilde H\right]=\cL\left[\vect w^j,H\right];
\ee
\item For any $n$-dimensional probability distribution $\vect r$, the thermo-Lorenz curve of the state
\be
\bar{\vect w}:=\sum_jr_j\tilde{\vect w}^j
\ee
under $\tilde H$ is given by
\be\label{lorcon}
\cL\left[\bar{\vect w},\tilde H\right]=\sum_jr_j\cL\left[\tilde{\vect w}^j,\tilde H\right]=\sum_jr_j\cL\left[\vect w^j,H\right].
\ee
\end{enumerate}
\end{alem}
\begin{proof}
Let $0\equiv s_0<s_1<s_2\dots<s_{D-1}<s_D\equiv1$ be the collection of the distinct absciss{\ae} at which the various Lorenz curves $\cL\left[\vect w^j,H\right]$ bend (or terminate). Note that $D\le n(d-1)+1=\cO(nd)$.

Let $\tilde H$ be a $D$-dimensional Hamiltonian with an energy spectrum $(\tilde E_1,\tilde E_2\dots,\tilde E_D)$ satisfying
\be
\tilde g_i\equiv\exp\left(-\beta\tilde E_i\right)=s_i-s_{i-1}
\ee
for all $i\in\{1,2\dots,D\}$, and let $\{\ket i\}$ be an orthonormal basis of associated eigenvectors.

For each $j\in\{1,2\dots,n\}$ and $i\in\{1,2\dots,D\}$, define
\be
\tilde{w}^j_i:=\cL\left[\vect w^j,H\right]\left(s_i\right)-\cL\left[\vect w^j,H\right]\left(s_{i-1}\right).
\ee
From the properties of the Lorenz curves $\cL\left[\vect w^j,H\right]$, it follows that
\be\label{lambj}
\forall j,\;\tilde{w}_1^j/\tilde g_1\ge\tilde{w}_2^j/\tilde g_2\dots\ge\tilde{w}_D^j/\tilde g_D. 
\ee
Eq.~(\ref{loreq}) follows. Furthermore, for any distribution $\vect r$, if $\bar{\vect w}=\sum_{j=1}^nr_j\tilde{\vect w}^j$, we have from (\ref{lambj})
\be
\bar{w}_1/\tilde g_1\ge\bar{w}_2/\tilde g_2\dots\ge\bar{w}_D/\tilde g_D. 
\ee
Hence, Eq.~(\ref{lorcon}) follows.
\end{proof}
\begin{prop}
Given a family $\left(\vect w^1,\vect w^2\dots,\vect w^n\right)$ of normalized states, a target state $\vect w$, and an $n$-dimensional probability distribution $\vect r$, the following two conditions are equivalent:
\begin{enumerate}
\item\label{eth} There exist TO $T^{(1)},T^{(2)}\dots,T^{(n)}$ such that
\be\label{vru}
\vect w=\sum_{j=1}^nr_jT^{(j)}\vect w^j.
\ee
\item\label{con2l} The Lorenz curves of the given states satisfy
\be\label{vrul}
\cL[\vect w]\le\sum_{j=1}^nr_j\cL[\vect w^j].
\ee
\end{enumerate}
\end{prop}
\begin{proof}[Proof $\Rightarrow$]
Assume that condition \ref{eth} holds. Consider the Lorenz curve of $\vect w$:
\be
\cL\left[\vect w\right]=\cL\left[\sum_{j=1}^nr_jT^{(j)}\vect w^j\right].
\ee
Lemma~\ref{lemlc} implies that
\be\label{tlcon}
\cL\left[\vect w\right]\le\sum_{j=1}^nr_j\cL\left[T^{(j)}\vect w^j\right].
\ee
But since each $T^{(j)}$ is a TO, the thermo-majorization condition for TO convertibility implies $\cL\left[T^{(j)}\vect w^j\right]\le\cL\left[\vect w^j\right]$.
\end{proof}
\begin{proof}[Proof $\Leftarrow$]
Assume that condition \ref{con2l} holds. First, by Lemma~\ref{lemlm}, for each $j$ there exists a TO $\tilde T^{(j)}$ that maps $\left(\vect w^j,H\right)$ to a $\left(\tilde{\vect w}^j,\tilde H\right)$ defined as in the lemma. Note that we have had to allow a change of Hamiltonian in this process. In a subsequent step we will be able to get back to the original Hamiltonian, resulting in an overall process that fits within our fixed-Hamiltonian formalism.

Next, we note using the same lemma that the state $\bar{\vect w}:=\sum_jr_j\tilde{\vect w}^j$ (under $\tilde H$) thermo-majorizes $\left(\vect w,H\right)$ by assumption. Therefore, there exists some TO $T$ mapping $(\bar{\vect w},\tilde H)$ to $(\vect w,H)$. Since the composition of two TO is also a TO, we can construct TO $T^{(j)}:=T\circ\tilde  T^{(j)}$ that satisfy condition \ref{eth}.
\end{proof}
Combining this result with Corollary~\ref{ptq} leads directly to Proposition~\ref{plqm} of the main matter. Proposition~\ref{plqm} implies that every instance of the CTO convertibility problem is the feasibility problem of a linear program of instance size $\cO(md)$, where $d$ is the dimensionality of S and $m$ that of Y. Converting the CTO problem statement (specified in terms of the matrices $U$ and $V$) to the corresponding linear program involves constructing the Lorenz curves of the $\vect u^{x}$, which in turn requires the calculation of each ${u}_i^{x}/g_i$, followed by sorting and interpolation; for $V$ we don't need the entire curves, only the positions of the bends. All these computations can be performed efficiently in practice; existing algorithms for linear optimization (of which feasibility problems are particularly simple cases) perform in time that scales cubically in the instance size. Therefore, overall, we have the convertibility condition in a form amenable to efficient computation.

The condition (\ref{lorcom}) runs over $s\in[0,1]$, but we don't really need to check for all values of $s$. Note that $\cL(\vect w)$ for any $\vect w$ is continuous and piecewise linear, with at most $d-1$ ``bends''. The possible horizontal coordinates (abscissae) where Lorenz curves bend are finite in number, and determined by $\vect g$. They form the set
\be
\sigma:=\left\{s=\sum_{i=1}^kg_{\pi(i)}:k\le d-1,\pi\in S_d\right\},
\ee
where $S_d$ is the group of permutations of $(1,2\dots,d)$. This set has size $|\sigma|=\tilde{d}-1$ with $\tilde{d}\le(d-1)d!+1$. Therefore, we have:
\begin{obs}\label{obedn}
In order to verify any instance of thermo-majorization relative to $\vect g$, we are required to compare the Lorenz curves of the two vectors only at the $\tilde d-1$ abscissae $s\in\sigma$.
\end{obs}
Index the elements of $\sigma$ as $\sigma_i$, such that $\sigma_1<\sigma_2\dots<\sigma_{\tilde{d}-1}$. Also define $\sigma_0:=0$ and $\sigma_{\tilde{d}}:=1$.

In addition to their continuity and piecewise linearity, which led to Observation~\ref{obedn}, Lorenz curves are also concave and monotonously non-decreasing. Therefore:
\begin{obs}\label{obens}
In order to determine whether $\vect u$ thermo-majorizes $\vect v$, we are required to compare the Lorenz curves of the two vectors only at the abscissae where $\cL\left[\vect v\right]$ bends.
\end{obs}
Now let $U$ and $V$ be two given resources. Similarly to $\sigma_i$, define $s_i$ ($i\in\{0,2\dots,D\}$) based on only the bends of $V$ (as described in the main matter). Note that $\{s_1,s_2\dots,s_D\}\subseteq\sigma$.

Consider some general set $\tilde\sigma=\left\{\tilde\sigma_1,\tilde\sigma_2\dots,\tilde\sigma_{\tilde{D}}\right\}$ of abscissae such that $\tilde{D}\ge D$ and $\{s_1,s_2\dots,s_D\}\subseteq\tilde\sigma$. In the following, we prove results that are valid for any choice of $\tilde\sigma$. Our proofs will therefore naturally apply to the special cases $\{s_i\}$ and $\sigma$.

\begin{prop}[General case of Proposition~\ref{lempq} of main matter] For any pair $(U,V)$, define the normalized bipartite distributions ${P}\equiv\left(\vect p^{1},\vect p^{2}\dots,\vect p^{\ell}\right)$ and ${Q}\equiv\left(\vect q^{1},\vect q^{2}\dots,\vect q^{m}\right)$ through
\begin{align}
p_{ix}&=\cL\left[\vect u^{x}\right](\tilde\sigma_i)-\cL\left[\vect u^{x}\right](\tilde\sigma_{i-1});\nonumber\\
q_{iy}&=\cL\left[\vect v^{y}\right](\tilde\sigma_i)-\cL\left[\vect v^{y}\right](\tilde\sigma_{i-1})
\end{align}
for $i\in\{1,2\dots,\tilde D\}$. Then,
\be
U\ct V\;\Longleftrightarrow\;\exists\textnormal{ row-stochastic }{R}:L{P}{R}\ge L{Q},
\ee
where the inequality is entriwise, and $L$ is the $\tilde D\times\tilde D$ matrix given by
\be
L=\left(\begin{array}{cccc}1&0&\dots&0\\
1&\ddots&\ddots&\vdots\\
\vdots&\ddots&\ddots&0\\
1&\dots&1&1
\end{array}\right).
\ee
\label{alpq}
\end{prop}
One of the corollaries of Proposition~\ref{lempq} was the following, which we now prove:
\begin{coro}[Corollary~\ref{coro5} of main matter]
Given athermality resources $(\vect u,\vect v)$ such that $\vect u\tho\vect v$ is possible, $U\equiv\left(p\vect u,[1-p]\vect g\right)$ can be converted to $V\equiv\vect v$ by CTO if and only if $p\ge p_{\min}$, where
\be
p_{\min}=\max_{i\in\{1\dots,\tilde D-1\}}\left[\fr{\cL[\vect v](\tilde\sigma_i)-\tilde\sigma_i}{\cL[\vect u](\tilde\sigma_i)-\tilde\sigma_i}\right].
\ee
\end{coro}
\begin{proof}
From Corollary~\ref{corom}, we have that $U\ct V$ if and only if
\be
p\cL[\vect u](\tilde\sigma_i)+(1-p)\cL\left[\vect g\right](\tilde\sigma_i)\ge\cL[\vect v](\tilde\sigma_i)\quad\forall i.
\ee
Noting that $\cL\left[\vect g\right](s)=s$, the above can be rephrased as
\be
p\left(\cL[\vect u](\tilde\sigma_i)-\tilde\sigma_i\right)\ge\cL[\vect v](\tilde\sigma_i)-\tilde\sigma_i\quad\forall i.
\ee
This leads to the claimed result.
\end{proof}

Using Proposition~\ref{alpq}, we find a connection with a relation called lower-triangular (LT) majorization \cite{MOA,NG15}. For two $\tilde D$-dimensional probability distributions $\vect p$ and $\vect q$,  we say $\vect p\ltM\vect q$ (``$\vect p$ LT-majorizes $\vect q$'') if there exists a ${\tilde D}\times {\tilde D}$ LT column-stochastic (LTCS) matrix $\Theta$ such that $\vect q=\Theta\vect p$. It can be shown easily that $\vect p\ltM\vect q$ is equivalent to $L\vect p\ge L\vect q$ componentwise. For this reason LT majorization is also called unordered majorization, alluding to the fact that usual majorization is defined similarly through partial sums but after the vector components have been reordered in nonincreasing order.

Coming to joint distributions, for a given ${R}$ the condition $L{P}{R}\ge L{Q}$ is equivalent to the condition that each column of ${P}{R}$ LT-majorize the corresponding column of ${Q}$. The condition ``there exists a row-stochastic ${R}$ such that $L{P}{R}\ge L{Q}$'' defines a preorder (reflexive and transitive binary relation) on the set of joint distributions with ${\tilde D}$ rows. Following Ref.~\cite{CUP15}, we will denote this as ${P}\rhd_\mathrm{c}{Q}$ (``${P}$ conditionally LT-majorizes ${Q}$''). For brevity, we denote by $\bbR^{{\tilde D}\times\ell}_{+,1}$ the set of all normalized ${\tilde D}\times\ell$ joint distributions. For a given ${P}\in\bbR^{{\tilde D}\times\ell}_{+,1}$, define
\be
\cM({P},k):=\left\{{Q}'\in\bbR^{{\tilde D}\times k}_{+,1}:{Q}'\lhd_\mathrm{c}{P}\right\},
\ee
which is called the markotope. Note that it is a compact convex set; this follows from the fact that the set of ${\tilde D}\times {\tilde D}$ LTCS matrices is convex and compact, as is the set of $\ell\times k$ row-stochastic matrices.
\begin{alem}
Given ${P}\in\bbR^{{\tilde D}\times\ell}_{+,1}$ and ${Q}\in\bbR^{{\tilde D}\times m}_{+,1}$,
\be
{P}\rhd_\mathrm{c}{Q}\iff\cM({P},k)\supseteq\cM({Q},k)\;\forall k\in\bbN.
\ee
\end{alem}
Let  $\cS_{\cM({P},m)}:\bbR^{{\tilde D}\times k}\to\bbR$ be the support function of the markotope, defined by
\be
\cS_{\cM({P},k)}(A):=\max\left\{\Tr(A^T{Q}'):{Q}'\in\cM({P},k)\right\}.
\ee
Support functions of non-empty compact convex sets have the property that $\cM({Q},k)\subseteq\cM({P},k)$ if and only if
\be
\cS_{\cM({Q},k)}(A)\leq\cS_{\cM({P},k)}(A)\;\forall A\in\bbR^{{\tilde D}\times k}_{+,1}.
\ee
From the last observation, the support function provides a characterization of conditional LT majorization. For a given ${P}\equiv\left(\vect p^{1},\vect p^{2}\dots,\vect p^{\ell}\right)$ and $A\equiv\left(\vect a^1,\vect a^2\dots,\vect a^k\right)$, the calculation of $\cS_{\cM({P},k)}(A)$ can be simplified as follows. Using the above insights on LT majorization, each ${Q}'\in\cM({P},k)$ can be written as $\left(\Theta^{(1)}{P}\vect{R}^1,\Theta^{(2)}{P}\vect{R}^2\dots,\Theta^{(k)}{P}\vect{R}^k\right)$, with ${R}\equiv\left(\vect{R}^1,\vect{R}^2\dots,\vect{R}^k\right)$ row-stochastic and each $\Theta^{(y)}$ LTCS. Therefore,
\begin{align}
\cS_{\cM({P},k)}(A)&=\max_{\vect\Theta,{R}}\sum_{y=1}^k\sum_{i,j=1}^{\tilde D}\sum_{x=1}^\ell a_{iy}\Theta^{(y)}_{ij}{p}_{jx}{R}_{xy}\nonumber\\
&=\sum_{x=1}^{\ell}\max_y\sum_{j=1}^{\tilde D}{p}_{jx}\max_{i\ge j}a_{iy}.
\end{align}
In the second line we used the structure of LTCS matrices. Note that $\max_{i\ge j}a_{iy}$ is a nonincreasing sequence in $j$. Therefore, it suffices to consider $A$ belonging to the set
\be
\bbR^{{\tilde D}\times k}_{+,1,\downarrow}=\left\{A\in\bbR^{{\tilde D}\times k}_{+,1}:\forall j,\;a_{1j}\ge a_{2j}\dots\ge a_{{\tilde D}j}\right\},
\ee
in which case for any ${\tilde D}$-dimensional distribution $\vect p$,
\be
\max_y\sum_{j=1}^{\tilde D}{p}_j\max_{i\ge j}a_{iy}=\max_y\vect p\cdot\vect a^y=:\omega_A(\vect p).
\ee
\begin{coro}
For ${P}\in\bbR^{{\tilde D}\times\ell}_{+,1}$ and ${Q}\in\bbR^{{\tilde D}\times m}_{+,1}$, ${P}\rhd_\mathrm{c}{Q}$ if and only if, for all $k\in\bbN$,
\be\label{pphip}
\sum_{x=1}^{\ell}\omega_A\left(\vect p^{x}\right)\ge\sum_{y=1}^{m}\omega_A\left(\vect q^{y}\right)\;\forall A\in\bbR^{{\tilde D}\times k}_{+,1,\downarrow}.
\ee
\end{coro}
\begin{obs}
Without loss of generality, we can restrict the above condition to the case of $k=m$. It is obvious that this case subsumes $k<m$. To see how it extends to $k>m$, consider some $A\equiv\left(\vect a^1,\vect a^2\dots,\vect a^k\right)\in\bbR^{{\tilde D}\times k}_{+,1,\downarrow}$. For each $y$, $\omega_A\left(\vect q^{y}\right)=\vect q^{y}\cdot\vect a^{f(y)}$ for some function $f(y)$. Define $B\equiv\fr1\alpha\left(\vect a^{f(1)},\vect a^{f(2)}\dots,\vect a^{f(m)}\right)\in\bbR^{{\tilde D}\times m}_{+,1,\downarrow}$, with $\alpha>0$ a suitable normalization factor. If (\ref{pphip}) holds for $k=m$, then $\sum_{x=1}^{\ell}\omega_B\left(\vect p^{x}\right)\ge\sum_{y=1}^{m}\omega_B\left(\vect q^{y}\right)$. But $\sum_{x=1}^{\ell}\omega_A\left(\vect p^{x}\right)\ge\alpha\sum_{x=1}^{\ell}\omega_B\left(\vect p^{x}\right)$, whereas $\sum_{y=1}^{m}\omega_A\left(\vect q^{y}\right)=\alpha\sum_{y=1}^{m}\omega_B\left(\vect q^{y}\right)$.\qed
\end{obs}
This immediately yields:
\begin{athm}[General case of Theorem~\ref{thmon} of main matter]
Let ${P}\in\bbR^{{\tilde D}\times\ell}_{+,1}$ and ${Q}\in\bbR^{{\tilde D}\times m}_{+,1}$. Then ${P}\rhd_\mathrm{c}{Q}$ if and only if for all matrices $A\in\bbR^{{\tilde D}\times m}_{+,1,\downarrow}$,
\be
\sum_{x=1}^{\ell}\omega_A\left(\vect p^{x}\right)\ge\sum_{y=1}^{m}\omega_A\left(\vect q^{y}\right).
\ee
\end{athm}
For a general choice of $\tilde\sigma$ that is independent of $U$ and $V$, the LHS and RHS above are conditional athermality monotones. Theorem~\ref{thmon} of the main matter is a special case of this theorem, where $\tilde{\sigma}$ is chosen to be just the set $\{s_i\}$. In this case, the quantities $\omega_A$ appearing in the theorem are $V$-dependent, and therefore the theorem is stated in terms of a conversion witness instead of monotones.

We close with the following result on asymptotic convertibility of states under CTO.
\begin{prop}[Proposition~\ref{asym} of the main matter]
Given $U\equiv\left(p_1\vect u^{|1},p_2\vect u^{|2}\dots,p_\ell\vect u^{|\ell}\right)$ and $V\equiv\left(q_1\vect v^{|1},q_2\vect v^{|2}\dots,q_m\vect v^{|m}\right)$, the conversion $U^{\otimes n_1}\ct V^{\otimes n_2}$ can be carried out \emph{asymptotically reversibly}, at the optimal rate
\be
\lim_{n_1\to\infty}\fr{n_2}{n_1}=\fr{{f}(U)}{{f}(V)},
\ee
where
$${f}(U):=\sum_{x=1}^\ell p_xF_\beta\left(\vect u^{|x}\right),$$
with $F_\beta(\vect u):=\sum_{i=1}^d{u}_i\left(E_i+\beta^{-1}\ln{u}_i\right)$ the \emph{free energy function}.
\end{prop}
\begin{proof}
We will use the properties of strongly typical (or letter-typical) sequences \cite{Shannon}, applied on $U\equiv U_{\rS\rX}$. Consider sampling the X part of the source $n_1$ times. Denote by $f_x$ the relative frequency of symbol $x$ in the resulting sequence. For any $\delta>0$, choose some $\epsilon\ge n_1/\delta^2$. Then, it is known from the theory of strong typicality that
\begin{equation}
\mathrm{Pr}\left[\left|f_x-p_x\right|\le\delta\sqrt{\fr{p_x(1-p_x)}{n_1}}\right]\ge1-\epsilon.
\end{equation}
By choosing $\delta(n_1)\in o\left(\sqrt{n_1}\right)$ and $\epsilon(n_1)\in o\left(n_1^0\right)$ such that $\epsilon(n_1)\in\omega\left(\delta^{-2}\right)$, we can make both $\epsilon$ and $\delta$ approach $0$ asymptotically (i.e., as $n_1\to\infty$). Therefore, in this limit, we can assume that $f_x=p_x$.

When the register X is in state $x$, the corresponding state of S is $\vect p^{|x}$. From the asymptotic resource theory of TO \cite{Reth}, it is known that any resource (i.e., any non-equilibrium state $\vect u\ne\vect g$) can be reversibly converted to a ``standard resource'' at a rate proportional to its free energy $F_\beta(\vect u)$; interconversion between arbitrary resources can be mediated by this standard resource. Similarly, under CTO, we can convert a joint state $U$ to $V$ by first converting copies of each conditional state $\vect p^{|x}$ to copies of the standard resource. The reversibility and the value of the conversion rate follow from the arguments in the previous paragraph, combined with the results of Ref.~\cite{Reth}.
\end{proof}

\end{document}